\documentclass[final,1p,times,nonatbib]{elsarticle}
\usepackage{graphicx}
\usepackage{amsmath,amssymb}
\usepackage{amsthm}
\usepackage{mathtools} 
\usepackage{enumitem}
\usepackage{bm}
\usepackage{bbm}
\usepackage[colorlinks=true, urlcolor=blue, linkcolor=blue, citecolor=blue]{hyperref}
\usepackage{tikz} 
\usetikzlibrary{matrix, decorations.pathreplacing, positioning}

\usepackage[%
style=phys,%
articletitle=true,biblabel=brackets,backend=biber,
chaptertitle=false,pageranges=false%
]
{biblatex}

\newcommand{\eeqref}{Eq.~\eqref}
\newcommand{\ccite}{Ref.~\cite}

\newcommand{\CC}{\mathbbm{C}}

\newcommand{\MM}{\mathbbm{M}}

\DeclareMathOperator{\tr}{Tr}

\DeclareMathOperator{\real}{Re}
\DeclareMathOperator{\imag}{Im}

\newcommand{\id}{\mathbbm{1}}
\newcommand{\idm}{\mathbbm{I}}

\newcommand{\M}{\mathcal{M}}

\newcommand{\bra}[1]{\langle #1 |}
\newcommand{\ket}[1]{| #1 \rangle}
\newcommand{\braket}[2]{\langle #1 | #2\rangle} 
\newcommand{\ketbra}[2]{| #1 \rangle\langle #2|}
\newcommand{\norm}[1]{\left\lVert #1 \right\rVert}

\newcommand{\Y}{\Upsilon}

\newtheorem{defn}{Definition}[section]

\newtheorem{obsr}[defn]{Observation}
\newtheorem{note}[defn]{Note}

\newtheorem{thm}[defn]{Theorem}
\newtheorem{lem}[defn]{Lemma}
\newtheorem{prop}[defn]{Proposition}
\newtheorem{cor}[defn]{Corollary}

\theoremstyle{remark}

\newtheorem*{pf}{Proof}
\newtheorem*{rmk*}{Remark}

\journal{Linear Algebra and its Applications}

\begin{document}
	
	\begin{frontmatter}
		
		
		
		\title{CP-preserving channels }
		\author{Indu Bala}
		\ead{indu.qic@gmail.com}
		\author{Sourav Das}
		\ead{souravmath.qic@gmail.com}
		\author{Swapan Rana}
		\ead{swapanqic@gmail.com}
		\address{Physics and Applied Mathematics Unit, Indian Statistical Institute, 203 B. T. Road, Kolkata 700108, India}

		\begin{abstract}
			Completely positive (CP) matrices are ubiquitous in modern science and technology with application in optimization, graph theory, and quantum entanglement. Recently, Johnston \emph{et al.} [Linear Algebra and its Applications, 2022] have cast CP matrices into the framework of quantum resource theories, where CP states serve as free states and CP-preserving channels act as free operations. This work addresses several questions raised in their work. Specifically, we provide the necessary and sufficient conditions of CP-preserving channels in small dimensions, which are necessary in higher dimensions, and discuss the resource quantification via the trace distance of non-negativity. By constructing an explicit counterexample, we demonstrate that the trace-distance measure of non-negativity violates strong monotonicity. We also provide an alternative proof that  every  CPDNN channel  $\Phi:\MM_n\to \MM_2$ is CPCP. Additionally, we show that any unital CPDNN map  $\Phi:\MM_2\to \MM_n$ is also CPCP. 

		\end{abstract}
		
		
	\end{frontmatter}
	
	
\section{Introduction}
A \emph{nonnegative} matrix is one with all its entries nonnegative, and will be denoted by $A\geq 0$. An $n\times n $ real symmetric matrix $A$ is doubly nonnegative (DNN) if it is  both nonnegative and positive semidefinite (PSD) --- that is, $\bra{x}A\ket{x}\geq 0$ for all $\ket{x}$, and will be denoted by $A \succcurlyeq 0$. Furthermore, $A$ is completely positive (CP) if it can be decomposed as $A=BB^T$, where B is a nonnegative matrix. While constructing a CP matrix is straightforward, verifying whether a given matrix is CP, is generally very difficult \cites{Dickinson.2014}[Chapter~3 of][]{Naomi}. Every completely positive matrix is DNN but  this condition is sufficient only for matrices upto size $4\times 4$. Very recently, Johnston \emph{et al.} \cite{JOHNSTON} have introduced the resource theory  of non-negativity of  quantum amplitude, where completely positive states serve as  free states. Within this framework, free operations are classified as $\mathcal{CP}$-preserving operations. They established the  necessary and sufficient conditions  for  CP-preserving  qubit channels. In this paper, we identify the  necessary and sufficient condition for $d=2,3,4$. We observe that the CP-preservinng conditions derived by Johnston \emph{et al.} can be recovered  in qubit case $(d=2)$. However  for dimensions greater than four, the provided conditions are only necessary.

To characterize the channels that remain  CP-preserving under an ancillary extension $(\idm_k\otimes \Phi)$, the authors in \ccite{JOHNSTON} have introduced completely positive completely positive (CPCP) maps. Choosing CPCP channels as the free operations satisfies the axiomatic framework of Chitambar and Gour \cite{Chitambar2019}, giving this resource theory a complete tensor product structure. Furthermore, the authors have established  equivalent forms for both  CPCP and completely positve doubly nonnegative (CPDNN) maps. Although CPDNN and CPCP maps coincide for $n=m=2$, their equivalence for $n=3, m=2$ remained an open problem. This question was recently settled in the affirmative by Cha \cite{cha2026}. In this paper, we provide an alternative proof showing that any CPDNN map \(\Phi \in \mathcal{L}(M_n, M_2)\) is also a CPCP map. In addition to this alternative proof, we demonstrate that any unital CPDNN linear map \(\Phi \in \mathcal{L}(M_2, M_n)\) is necessarily a CPCP map.

A central objective in  the framework of quantumk resource theories is identifying whether a given theory admits a maximally resourceful state --- i.e., somewhat a \emph{golden standard} state from which any other states can be prepared using only free operations. For qubit systems ($d=2$), it was established that for any qubit state $\rho$,  there exists a CP-preserving quantum channel $\Phi$ such that $\Phi(\gamma)=\rho$ where $\gamma$ is the density matrix $\gamma=(\idm \pm \sigma_y)/2$. However, the existence of such aresource state for higher dimension $(d\geq 3)$ was left open. Although we have discussed about the possibility of its non-existence in higher dimensions, we have not been able to find a proof yet.  However, we have fully answered a related question raised in \cite{JOHNSTON}, namely, the trace-distance of non-negativity fails to satisfy strongly monotonicity, by providing an explicit counterexample. 
 
 This article is organized as follows. Section~\ref{Sec: CP-Cond} characterizes CP-preserving channels for dimensions $d\leq 4$, in terms of the channel, its dual, and its Choi state, mainly leveraging the fact that in these dimensions CP and DNN are the same. This extends the characterization result for $d=2$ from \cite{JOHNSTON}.  CPCP and CPDNN channels are analyzed in Section~\ref{CPDNN}, which answers a question raised in \cite{JOHNSTON}. In Section~\ref{Sec:Sym.Reduction}, we consider evaluation of several measures of non-negativity, if the states under consideration has some symmetries. We show that in such cases, the nearest CP states can be found efficiently, in small dimensions ($d\leq 4$) even without any computation.   Section~\ref{Sec:TDNN} provides an example demonstrating that the trace distance of non-negativity lacks strong monotonicity, thereby settling another question raised in \cite{JOHNSTON}.

\section{Characterization of CP-preserving channel for $2 \leq d \leq 4$\label{Sec: CP-Cond}} 

In this Section we will fully characterize the CP-preserving channels for $2 \leq d \leq 4$. We will give several equivalent conditions and a semi-definite program (SDP). Note that this dimension restriction is natural, as for $d\geq 5$ the CP matrices are unwieldy. However, we will also present some necessary conditions for higher dimensions. We first derive the conditions in terms of the dual map of a channel.

For any linear map $\Phi:\MM_{d_1}(\CC)\to \MM_{d_2}(\CC)$, the corresponding dual map $\Phi^*:\MM_{d_2}(\CC)\to \MM_{d_1}(\CC)$ is uniquely defined via the Hilbert-Schmidt inner product identity:
\begin{equation}\label{Def:Dual.Channel}
	\braket{\Phi(X)}{Y}_{HS}=\braket{X}{\Phi^*(Y)}_{HS},\quad \forall\,X\in \MM_{d_1}(\CC),\, Y\in \MM_{d_2}(\CC).	
\end{equation}

Specifically, for a channel $\Phi=\{A_k\}$, i.e., if  the channel $\Phi$ is represented in its Kraus form as 
\[
\Phi(X) = \sum_k A_kXA_k^\dagger, \quad \text{with} \quad \sum_k A_k^\dagger A_k = \idm_{d_1},
\] 
then its dual map (not necessarily a channel as it need not be trace preserving) is $\Phi^*=\{A_k^\dagger\}$, having the Kraus form \cite{CHOI1975285,Watrous2018}
\[
\Phi^*(Y)=\sum_k A_k^\dagger Y A_k.
\]

In terms of  matrix elements, $\Phi$ and its dual  related as
\begin{equation}\label{eq-elements-relation}
	\Phi(\ketbra{i}{j})_{s,r}=\Phi^*(\ketbra{r}{s})_{j,i}
\end{equation}
We utilize this dual structure to explore channels that preserve CP matrices. By definition, a channel $\Phi$ is CP-preserving if it maps CP states exclusively to CP states. While some structure of the linear maps preserving CP matrices were known in the literature,   \cite[see e.g., pp.~514-515, ][]{Naomi}, recently Johnston and Sikora \cite{JOHNSTON} have studied them from the perspective of quantum channels. Employing results from quantum information, they have completely characterized the qubut case ($d=2$). In this Section, we present necessary and sufficient conditions for $d\leq 4$, along with a necessary condition applicable to higher dimensions.

\begin{thm}\label{thm-DNN}
	A channel  $\Phi:\MM_{d_1}(\CC)\to \MM_{d_2}(\CC)$, $2\leq d_i\leq 4$, is CP-preserving if and only if $\bra{x}\Phi^*\left(\ket{i}\bra{j}\right)\ket{x}\geq 0$, for all indices $i,j,$ and $\ket{x}\geq0$.
\end{thm}

\begin{proof}
	Consider a pure state $\ket{x}\bra{x}$, where $\ket{x}\geq0$. For every such $\ket{x}$, there exists a matrix $T^x$ such that
	\begin{equation}\label{Eq:CP.T}
		\Phi(\ket{x}\bra{x})=T^x.
	\end{equation}
	Now the channel condition implies $T^x$ is PSD. Since $\mathcal{CP}=\mathcal{DNN}$ up to dimension $4$, it follows that
	\[ T^x \geq 0 \iff \Phi \text{ is CP-preserving}.\]
	The claim follows by calculating the matrix element of $\Phi^*$ from Eqs.~\eqref{Eq:CP.T} and \eqref{Def:Dual.Channel}, 
	\[
	T^x_{i,j} = \bra{i}\Phi\left(\ket{x}\bra{x}\right)\ket{j} = \bra{x}\Phi^*\left(\ket{i}\bra{j}\right)\ket{x}.
	\]
\end{proof}

In order to explore this condition further, we 	use the notion of copositive matrices. 
\begin{defn}[Copositive]
	A real symmetric matrix $A$  is said to be  copositive if $\bra{x}A\ket{x}\geq 0$ for all $\ket{x}\geq 0$.
\end{defn} 
\begin{obsr}\label{Obs:Co.Phi}
	The condition $\bra{x}\Phi^*\left(\ket{i}\bra{j}\right)\ket{x}\geq 0$ for all $\ket{x}\geq 0$ in Theorem~\ref{thm-DNN} is equivalent to  the symmetric part [i.e., writing $X:=X_{\text{sym}}+X_{\text{skew-sym}}=(X+X^T)/2+(X-X^T)/2$] of $\Phi^*\left(\ket{i}\bra{j}\right)$  being copositive. Or equivalently, the symmetric part of the real part [from the Cartesian decomposition $X:=\real (X)+i\imag (X)$] of $\Phi^*\left(\ket{i}\bra{j}\right)$ being copositive and its imaginary part being skew-symmetric.
\end{obsr}

To see this explicitly, let us denote   $C^{sr} \coloneqq \Phi^*(\ket{s}\bra{r})=A^{sr}+i B^{sr}$,
where $A^{sr}$ and  $B^{sr}$ are real and imaginary part of $C^{sr}$ respectively.  
Now  \[\bra{x}\Phi^*(\ket{s}\bra{r})\ket{x}\geq0\]
\[	\implies
\bra{x}C^{sr}\ket{x}\geq 0, \quad \forall\, x\geq0\]
\[\implies  \bra{x}B^{sr}\ket{x}= \bra{x}\frac{B^{sr}+B^{{sr}^T}}{2}\ket{x}=0, \quad \forall \, x\geq 0 \]
\[\implies B^{sr}=-B^{{sr}^T}.\]
The equivalent condition for other decomposition follows from simialr arguments.

The characterization above was in terms of $\Phi^*$. We can give the equivalent characterization in terms of $\Phi$ alone.
\begin{obsr}\label{obsr-CP}
	A channel $\Phi:\MM_{d_1}(\CC)\to \MM_{d_2}(\CC)$, $2\leq d_i\leq 4$, is CP-preserving if and only if \[\frac{1}{2}\sum_{i,j=1}^{d_1}\Phi(\ket{i}\bra{j})_{r,s}\left(\ket{i}\bra{j}+\ket{j}\bra{i}\right),\] is copositive for all $r,s=1,2,\dotsc,d_2$, where $X_{r,s}$ is the $rs$-th entry of the matrix $X$: $X_{r,s}:=\bra{r}X\ket{s}$.
\end{obsr}
\begin{proof}
	If we assume the channel $\Phi$ to be CP-preserving, by Theorem~\ref{thm-DNN}, $\bra{x}\Phi^*(\ket{s}\bra{r})\ket{x}\geq0$. So by Observation~\ref{Obs:Co.Phi}, $\frac{C^{sr}+C^{{sr}^T}}{2}$ is copositive. Thus, 
	\[\frac{\Phi^*(\ket{s}\bra{r})+\Phi^*(\ket{s}\bra{r})}{2}^T=\frac{C^{sr}+C^{{sr}^T}}{2}=\frac{1}{2}\sum_{i,j}\Phi(\ket{i}\bra{j})_{r,s}\left(\ket{i}\bra{j}+\ket{j}\bra{i}\right)\]
	will be copositive for all $r,s$. We have used the fact that $C^{sr}_{i,j}=\Phi(\ket{j}\bra{i})_{r,s}$.
	
	To prove the  other side,  let us  assume  $\frac{1}{2}\sum_{i,j}\Phi(\ket{i}\bra{j})_{r,s}\left(\ket{i}\bra{j}+\ket{j}\bra{i}\right)$  is copositive for all $r,s$,
	and $\Phi$ is a channel. 
	Now \[\frac{1}{2}\sum_{i,j}\Phi(\ket{i}\bra{j})_{r,s}\left(\ket{i}\bra{j}+\ket{j}\bra{i}\right)=\frac{\Phi^*(\ket{s}\bra{r})+\Phi^*(\ket{s}\bra{r})}{2}^T\] 
	is copositve. Noticing that $\bra{x}X^T\ket{x}=\tr\left(X^T\ketbra{x}{x}\right)=\tr\left[\left(\ketbra{x}{x}\right)^TX\right] = \tr\left(\ketbra{x}{x}X\right)=\bra{x}X\ket{x}$ for all $\ket{x}\geq 0$, the above relation
	
	\[\implies \bra{x} \Phi^*(\ket{s}\bra{r})\ket{x}\geq 0,  \quad\forall\, \ket{x}\geqq 0, \& \, r,s.\]
	Hence, by Theorem~\ref{thm-DNN}, the channel $\Phi$ is CP-preserving.
	
	As an example, a channel  $\Phi:\MM_{3}(\CC)\to \MM_{d}(\CC)$, $2\leq d\leq 4$, is CP-preserving if and only if the following matrices
	\[Y_{r,s}=\begin{pmatrix}
		\Phi(\ket{1}\bra{1})_{r,s} & 	\Phi(\ket{1}\bra{2}+\ket{2}\bra{1})_{r,s}  &	\Phi(\ket{1}\bra{3}+\ket{3}\bra{1})_{r,s}\\
		\Phi(\ket{1}\bra{2}+\ket{2}\bra{11})_{r,s} & 	\Phi(\ket{2}\bra{2})_{r,s} &\Phi(\ket{2}\bra{3}+\ket{3}\bra{2})_{r,s} \\
		\Phi(\ket{1}\bra{3}+\ket{3}\bra{1})_{r,s}  & \Phi(\ket{2}\bra{3}+\ket{3}\bra{2})_{r,s} &\Phi(\ket{3}\bra{3})_{r,s}
	\end{pmatrix}\]
	are  copositive for all $r,s=1,2,\dotsc,d$. 
	
	We notice that while these are necessary and sufficient conditions, not all of them are independent (as a constraint). We also note that for $d\leq 4,$ easily verifiable necessary and sufficient condition are known for copositivity \cites{Diananda.1962}[Chapter~2.3 of ][]{Naomi}, thereby checking CP-preserving for $d\leq 4$ becomes an easy problem (albeit, writing them explicitly in terms of the channel parameters, like that for the $d=2$ case in \cite{Johnston2019}, is perhaps neither possible, nor desirable). 
	\begin{cor}
		Unital CP-preserving quantum channel will always be maximally incoherent operation (MIO).
	\end{cor}
	\begin{pf}If $\Phi$ is CP-presrving unital quantum channel then by Observation~\ref{Obs:Co.Phi} diagonal elements of  $\Phi^*(\ket{i}\bra{j})$ will always be non negative but trace preserving condition  of  $\Phi^*$ will imply diagonal elements 
		of $\Phi^*(\ket{i}\bra{j})$ will be zero $\forall$ $i\neq j$. Hence  $\Phi$ will map diagonal  to diagonal follows from Eq.~\eqref{eq-elements-relation}. So $\Phi$ will be MIO.
	\end{pf}
	
	\begin{cor}
		Off diagonal entry of Image of symmetric computational basis under unital CP-preserving quantum channel will always be non negative.
	\end{cor}
	
	\begin{pf}\[\Phi\left(\ket{i}\bra{j}+\ket{j}\bra{i}\right)=\Phi\left(\ket{i}\bra{i}+\ket{i}\bra{j}+\ket{j}\bra{i}+\ket{j}\bra{j}\right)-\Phi\left(\ket{i}\bra{i}+\ket{j}\bra{j}\right)\]
		Right hand side is difference of two CP matrices so MIO condition will imply  off diagonal entry  of $\Phi\left(\ket{i}\bra{j}+\ket{j}\bra{i}\right)$ will be non negative.
	\end{pf}	
	\begin{note}
		If we assume 
		\[\Phi(\ket{i}\bra{j})=R_{i,j}+i S_{i,j}\]
	\end{note}
	For $i\neq j,\,S_{i,j}$ is  symmetric and $S_{i,i}=0$.
	We can write every symmetric matrix as difference of two CP
	matrices. So if a map is CP-preserving then it will also symmetric preserving.
\end{proof}
Next we will give CP-preserving condition in term of Choi state.
\begin{obsr}{Channel $\Phi$ is CP-preserving if and only if $J_\Phi+J_\Phi^{\Gamma_2} =\sum\limits_{r,s}Y_{r,s}\otimes E_{r,s}$.} where each $Y_{r,s}$ is copositive, and $\{E_{i,j}:=\ketbra{i}{j}\}$ is the standard matrix basis.
\end{obsr}
\begin{proof}
	Let us assume $\Phi$ is channel and $J_\Phi$ is its Choi matrices so 
	\[J_{\Phi}=\frac{1}{d}\sum_{i,j} \ket{i}\bra{j}\otimes \Phi(\ket{i}\bra{j})\]
	\[\frac{1}{2}\left(J_{\Phi}+J_{\Phi}^{\Gamma_2}\right)=\frac{1}{2d}\sum_{i,j} \ket{i}\bra{j}\otimes \Phi(\ket{i}\bra{j}+\ket{j}\bra{i})\]
	\[=\frac{1}{2d}\sum_{i,j,r,s} \ket{i}\bra{j}\otimes  \Phi(\ket{i}\bra{j}+\ket{j}\bra{i})_{r,s}\ket{r}\bra{s}\]
	
	\[=\frac{1}{d}\sum\limits_{i,j,r,s} \Phi\left((\ket{i}\bra{j}+\ket{j}\bra{i}\right )_{r,s} \ket{i}\bra{j} \otimes \ket{r}\bra{s}\]
	\[=\sum_{r,s} Y_{r,s}\otimes \ket{r}\bra{s}\]
	By observation \ref{obsr-CP}, $\Phi$ is CP-preserving if and only if $Y_{r,s}=\frac{1}{2}\sum\limits_{i,j}\Phi(\ket{i}\bra{j}+\ket{j}\bra{i})_{r,s}\left(\ket{i}\bra{j}\right)$ is copositive for all r,s.
	
\end{proof}

\subsection{Recover of qubit CP-preserving channel conditions}
Now we will recover the CP-preserving condition for the qubit case, reported in \ccite{JOHNSTON}.

By observation~\ref{obsr-CP},  qubit channel  $\Phi$ is CP-preserving if and only if

\[Y_{r,s}=\begin{pmatrix}
	\Phi(E_{1,1})_{r,s} & 	\Phi(E_{1,2}+E_{2,1})_{r,s} \\
	\Phi(E_{2,1}+E_{1,2})_{r,s} & 	\Phi(E_{2,2})_{r,s}  			
\end{pmatrix}\]
is copositive for all $r,s=1,2$, which are equivalent to
\[\Phi(E_{1,1})_{r,s} \geq 0, \quad\forall r,s,\]
\[\Phi(E_{2,2})_{r,s} \geq 0, \quad\forall r,s,\]
\[\Phi(E_{1,2}+E_{2,1})_{r,s}+\sqrt{\Phi(E_{1,1})_{r,s}\Phi(E_{2,2})_{r,s}} \geq 0, \quad\forall r,s.\]
If we assume matrix representation of the qubit channel  $\Phi$ in Pauli basis as
\[[\Phi]=\begin{pmatrix}
	1 & 0 & 0 & 0\\
	t_x & T_{x,x} & T_{x,y} & T_{x,z}\\
	t_y & T_{y,x} & T_{y,y} & T_{y,z}\\
	t_z & T_{z,x} & T_{z,y} & T_{z,z}\\
\end{pmatrix}\]
Then
\[\Phi(E_{1,1})=\frac{1}{2}\Phi(\id +\sigma_z)=\frac{1}{2}\begin{pmatrix}
	1+t_z+T_{z,z} & t_x +T_{x,z}- i( t_y +T_{y,z}) \\
	t_x+T_{x,z}+i (t_y+T_{y,z})& 1-t_z-T_{z,z}
\end{pmatrix}\]
\[\Phi(E_{2,2})=\frac{1}{2}\Phi(\id -\sigma_z)=\frac{1}{2}\begin{pmatrix}
	1+t_z-T_{z,z} & t_x -T_{x,z}- i( t_y -T_{y,z}) \\
	t_x-T_{x,z}+i (t_y-T_{y,z})& 1-t_z+T_{z,z}
\end{pmatrix}\]

\[\Phi(E_{1,2})=\frac{1}{2}\Phi(\sigma_x +i\sigma_y)=\frac{1}{2}\begin{pmatrix}
	T_{z,x}+i T_{z,y} &  T_{x,x}-i T_{y,x}+i T_{x,y}+T_{y,y}\\
	T_{x,x} +i T_{y,x}+i T_{x,y}-T_{y,y}& -T_{z,x}-i T_{z,y}
\end{pmatrix}\]
Hence 
\[\Phi(E_{1,2}+E_{2,1})=\begin{pmatrix}
	T_{z,x} &  T_{x,x}-i T_{y,x}\\
	T_{x,x}+i T_{y,x}& -T_{z,x}
\end{pmatrix}\]

Since 
\[\Phi(E_{1,1})_{r,s} \geq 0 \quad \text{and}\quad  \Phi(E_{2,2})_{r,s} \geq 0  \quad\forall r,s\]
So
$\Phi(E_{i,i})_{j,j}\geq 0$ for $i,j=1,2.$

\[   \implies 1\pm t_z \geq |T_{z,z}|\]
This condition is  superfluous as it is already covered in channel condition . \\
$\Phi(E_{i,i})_{1,2}\geq 0$
\[ \implies t_x\geq  |T_{x,z}|\,\text{and} \,t_y=T_{y,z}=0\]

$\Phi(E_{1,2}+E_{2,1})_{r,s}+\sqrt{\Phi(E_{1,1})_{r,s}\Phi(E_{2,2})_{r,s}} \geq 0 \quad\forall r,s$ 
\[\implies T_{y,x}=0\]
and 
\[ -\sqrt{(1+t_z)^2-T_{z,z}^2}\leq  T_{z,x}\leq \sqrt{(1-t_z)^2-T_{z,z}^2} \]

\[T_{x,x}\geq -\sqrt{t_x^2-T_{x,z}^2}\]
Finally when we consider all inequality together we will get
\[t_y=T_{y,z}=T_{y,x}=0\]
\[|T_{x,z}|\leq t_x\]
\[-\sqrt{(1+t_z)^2-T_{z,z}^2}\quad\leq T_{z,x}\leq \sqrt{(1-t_z)^2-T_{z,z}^2}\]
\[T_{x,x}\geq -\sqrt{t_x^2-T_{x,z}^2}\]

\subsection{Semidefinite Program for CP-preserving Channel Estimation ($2\leq d\leq 4$)}	
Given a set of fixed input states $\rho_i$ and corresponding output states $\sigma_i$, we address the state-transformation problem: does there exist a CP-preserving quantum channel $\Phi$ such that $\Phi(\rho_i) = \sigma_i$ for all $i$ \cite{2023semidefinite} ? 

Utilizing the Choi–Jamiołkowski isomorphism, we express the channel in terms of its Choi matrix $J$. The problem of finding a valid channel can be cast as the following Semidefinite Program (SDP) CVXProgram \cite{cvx,boyd2004convex} :

\begin{equation}
	\begin{aligned}
		\text{find} \quad & J \\
		\text{subject to} \quad & \tr_1 \left[ J \left(\rho_i^T \otimes \id_d\right) \right] = \sigma_i, \quad \forall i, \\
		& J \succcurlyeq 0, \\
		& \tr_2(J) = \id_d, \\
		& \frac{C^{sr} + {C^{sr}}^T}{2} \in \mathcal{CP}^*, \quad \forall s,r \in \{1, \dots, d_2\},
	\end{aligned}
\end{equation}
where the entries of the matrix $C^{sr} \in \mathcal{M}_{d_2}(\CC)$ are extracted from the Choi matrix $J$ according to the index relation:
\[
\left(C^{sr}\right)_{j,i} = J_{(i-1)d+r, \, (j-1)d+s}
.\]
Here, the action of the channel on an arbitrary state $X$ is explicitly evaluated from the Choi matrix via $\Phi(X) = \tr_1 \left(J \left(X^T \otimes \id_d\right) \right)$.
\begin{prop}
	There exists no unital CP-preserving channel $\Phi$ such that:
	\[
	\Phi(\ket{\psi}\bra{\psi}) = \ket{\phi}\bra{\phi},
	\]
	where the respective state vectors are given by:
	\[\ket{\psi}=\frac{1}{2}\begin{pmatrix}
		1 & i & -1 & -i
	\end{pmatrix}^T\]
	\[\ket{\phi}=\frac{1}{2}\begin{pmatrix}
		1 & i & i & -1
	\end{pmatrix}^T\]
\end{prop}

\begin{proof}
	This assertion is established numerically by utilizing the semidefinite programming (SDP) formulation developed in the previous section. By adding the unital constraint $\tr_1(J) = \id_{d_2}$ to the feasibility program and executing the optimization via MATLAB , the primal problem is found to be strictly infeasible. Thus, there exists no such CP-preserving unital channel.
\end{proof}

\section{CPDNN channel is equivalent to CPCP channel for $\mathcal{L}(\M_n,\M_2 )$ \label{CPDNN}}
A CPCP  map is a  linear map that preserves the CP of matrices, not only when acting  locally, but also when tensored with an  identity map of arbitray dimension. The structural characterization of CPCP linear maps in term of their Choi states and Kraus operators is provided in the theorem below. 
\begin{thm}\cite[Theorem~1]{JOHNSTON}\label{thm:CPCP_choi}
	Suppose $\Phi \in \mathcal{L}(M_n,M_m)$. The following are equivalent:
	\begin{enumerate}
		\item[(i)] $\Phi$ is CPCP  .
		
		\item[(ii)] $J_\Phi$ is CP matrix 
		
		\item[(iii)] There will exist at least one Kraus representation in which Kraus operators have non - negative entries. 
	\end{enumerate}
\end{thm} 
Furthermore, the characterization of CPDNN channel is given as
\begin{thm}\cite[Theorem~4]{JOHNSTON}\label{thm:DNN_choi}
	Suppose $\Phi \in \mathcal{L}(M_n,M_m)$. The following are equivalent:
	
	\begin{enumerate}
		\item[(i)] $\Phi$ is CPDNN. 
		
		\item[(ii)] $J_{\Phi}$ is DNN matrix. 
	\end{enumerate}
\end{thm}
\begin{lem}\label{-DNN}
	Corresponding to quantum operation $\Phi$, Choi matrix  $J_\Phi$ is DNN  if and only if $J_{\Phi^*}$ is DNN.
\end{lem}

\begin{pf}
	Let us assume $\Phi^*$ is the dual map corresponding to the quantum operation $\Phi$.
	If $\Phi$ is quantum operation then
	$J_\phi$ is positive semidefinite if and only if $J_{\phi^*}$ is positive semidefinite.
	Now by Eq.~\eqref{eq-elements-relation}
	\[\Phi(E_{i,j})_{s,r}=\Phi^*(E_{r,s})_{j,i}\]
	which implies  $J_{\Phi}$ is non negative if and only if $J_{\Phi^*} $. So $J_\Phi$ is DNN if and only if $J_{\Phi^*}$ is DNN.
\end{pf}

\begin{lem}\label{CPCP}
	Corresponding to quantum operation $\Phi$, Choi matrix  $J_\Phi$ is CP matrix  if and only if $J_{\phi^*}$ is CP.
\end{lem}
\begin{pf}
	If $\Phi$ is CPCP then there will exist atlest one set of Kraus operators say $\{A_i\}$ such that each $A_i$ is non negative(entrywise) \cite[Theorem~1]{JOHNSTON} such that
	\[\Phi(X)=\sum_i A_iXA_i^T\]
	The corresponding dual map \(\Phi ^{*}\) is given by:
	\[\Phi^*(Y)=\sum_iA_i^TYA_i\]
	
	So $\{A_i^T\} $ will be the  Kraus operator for $J_{\Phi^*}$. Hence $J_{\Phi}$ is CP if and only if $J_{\Phi^*}$ is CP.
\end{pf}
\begin{thm}\cite{Berman}
	If A is symmetric and nonnegative and if its comparison
	matrix M(A) is positive semidefinite, then A is CP.
\end{thm} 
\begin{thm}\label{CPDNN,CPCP,2}
	CPDNN channel is equivalent to CPCP channel for $\mathcal{L}(\M_n,\M_2 )$.
\end{thm}
\begin{pf}Let $\Phi$ be a CPDNN channel then its Choi matrix $J_\Phi$ will be DNN. Trace preserving condition implies \[\Phi(E_{i,j})_{m,m}=0\,\text{for}\,i\neq j\]
	Now from Eq.~\ref{eq-elements-relation} \[\Phi(E_{i,j})_{m,m}=\Phi^*(E_{m,m})_{j,i}=0\,\forall i\neq j\] 
	Hence \[J_{\Phi^*}=\begin{pmatrix}  
		D&F\\
		F^T&I-D
	\end{pmatrix}\]If $J_\Phi$ is DNN so $J_{\Phi^*}$ is DNN. The comparison matrix of $J_{\Phi^*}$ is positive semidefinite, so $J_{\Phi^*}$ is CP then from lemma(\ref{CPCP}) $J_\Phi$ will also CP.\\
\end{pf}
\begin{rmk*}
	\[X=	
	\begin{pmatrix}
		A_{11}& A_{12} & \hdots &  A_{1n}\\
		A_{12}^\dagger & A_{22} & \hdots & A_{2n}\\	
		\vdots & \vdots & \ddots & \vdots \\
		A_{1n}^\dagger&   A_{2n}^\dagger &\hdots & A_{nn}
	\end{pmatrix}
	\]
\end{rmk*}

Given that $X$ is a positive semidefinite matrix with entrywise nonnegative elements, where each block $A_{ij}$ is of size $2 \times 2$ satisfying $\text{Tr}(A_{ij})=0$ for all $i < j$ and $\text{Tr}(A_{ii})=1$ for all $i$, it follows from Theorem~\ref{CPDNN,CPCP,2} that $X$ is CP.
\begin{thm}\label{CPDNN,CPCP,Unital}
	Every CPDNN unital linear map $\mathcal{L}(\M_2,\M_n )$ are CPCP.
\end{thm}
\begin{pf} On the contrary if there exist CPDNN unital linear map $\Phi$ which is not CPCP then corresponding dual map  $\Phi^*:\mathcal{M}_n\to\mathcal{M}_2$ will be CPDNN quantum chaannel but not CPCP but by above theorem it is  not possible. Hence we will get contradiction. So every CPDNN unital linear map from $\mathcal{M}_2\to\mathcal{M}_n$ will be CPCP.
	
\end{pf}

\begin{rmk*}
	\[Y=	
	\begin{pmatrix}
		B_{11} & B_{12}\\
		B_{12}^\dagger  & \idm_n - B_{11}
	\end{pmatrix}
	\]\end{rmk*}
Given that $Y$ is a PSD matrix with entrywise nonnegative elements where each block $B_{ij}$ is of size $n \times n$, it follows from Theorem~\ref{CPDNN,CPCP,Unital} that $Y$ is CP.

\section{Evaluation of measures of non-negativity: symmetry reductions \label{Sec:Sym.Reduction}}

Following the general framework of a QRT, many quantifiers (or measures) of non-negativity were introduced in \ccite{JOHNSTON}. The two main quantifiers studied there (the $1$-norm of negativity measure was defined only for pure states, presumably, extended to the mixed states via usual convex roof construction), namely, the robustness ($N^R$) and the trace distance ($N^T$) can be efficiency computed using the SDP for $d\leq 4$, and for $d\geq 5$ can be approximated via the inequalities
\begin{equation}\label{Eq:Higherarchy.N}
N^X_{\mathcal{DNN}}(\rho)\leq N^X_{\mathcal{CP}}(\rho)\leq N^X_{\mathcal{DDN}}(\rho), \quad X=R, T.
\end{equation}

In this Section, we show that if $\rho$ has some symmetries, then the nearest CP state can be taken as the one having some special form, thereby reducing the computation. We will focus on the computation for the uniform state, considered in detail as example~1 in \ccite{JOHNSTON}, defined by 
\begin{equation}\label{Def:Uniform.d}
\ket{\Y_d}:=\sum\limits_{i=1}^d \omega^{i-1}\ket{i},\quad \omega:=e^{2\pi i/d},
\end{equation} 
and for the measure $N^T$. The result for $N^T(\ket{\Y})$ will also be used in the following Sections. 

Now, noticing that the only free (unitary) symmetries that keep a state invariant in \mbox{non-negativity} are the permutations, we have the following result.
 \begin{prop}\label{Prop:Sym.Reductions}
 	Let a state $\rho$ be invariant under a set of $n$ permutations $\{P_i\}$. If $\tau$ is a closest CP state to $\rho$ for evaluating $N^T(\rho)$, then $\sigma:=(\sum P_i\tau P_i^T)/n$ is also a closest CP state.
 \end{prop}
  \begin{proof}
 	Since $\tau$ is the closest CP state and $\sigma$ is another CP state (being a convex mixture of CP states), we must have \begin{equation}\label{Eq:NT.Leq.Rho-Sigma}
 		N^T(\rho):=\norm{\rho-\tau}_1\leq \norm{\rho-\sigma}_1.
 	\end{equation}
 	However, using the invariance of $\rho$ (and trace norm) under permutations and triangle inequality, we have 
 	\begin{equation}\label{Eq:NT.Geq.Rho-Sigma}
 		\norm{\rho-\sigma}_1= \,\norm{\rho-\frac{1}{n}\sum\limits_{i=1}^nP_i\tau P_i^T} =\, \norm{\frac{1}{n}\sum\limits_{i=1}^nP_i\left(\rho-\tau\right)P_i^T}_1\leq \frac{1}{n}\sum\limits_{i=1}^n\norm{P_i\left(\rho-\tau\right)P_i^T}_1 = \norm{\rho-\tau}_1=: N^T(\rho).
 	\end{equation}
 	The claim follows from \eeqref{Eq:NT.Leq.Rho-Sigma} and \eeqref{Eq:NT.Geq.Rho-Sigma}.
 \end{proof}

With only a slight modifications in the arguments, it readily follows that the Proposition~\ref{Prop:Sym.Reductions} is extendable to other domains, as well as other measures. Thus we have the following results.

\begin{obsr}\label{Obs:Sym.Reductions.R.T}
		Let a state $\rho$ be invariant under a set of $n$ permutations $\{P_i\}$. If $\tau$ is a closest CP state to $\rho$ for evaluating $N^X_Y(\rho)$, then $\sigma:=(\sum P_i\tau P_i^T)/n$ is also a closest CP state, for all $X=R,T$ and $Y=\mathcal{DDN},\mathcal{CP},\mathcal{DNN}$.
\end{obsr}

We now concentrate on evaluating $N^T(\ket{\Y_d})$. Noticing that the state $\rho=\ketbra{\Y_d}{\Y_d}$ remains invariant under $n$ number of permutations (whose form are also obvious), where \[n=\begin{cases}
	\frac{d-1}{2}, & \text{for odd } $d$,\\
	\frac{d}{2}, & \text{for even } $d$,
	\end{cases}\]
the closest state $\sigma$ for evaluating $N_Y^T(\ket{\Y_d})$, $Y=\mathcal{DDN},\mathcal{CP},\mathcal{DNN}$ can be taken as 
\begin{equation}\label{Eq:Closest.sigma.N.T}
\sigma=\begin{tikzpicture}[baseline=(m.center)]
	\matrix (m) [
	matrix of math nodes,
	left delimiter={[},
	right delimiter={]},
	nodes in empty cells,
	row sep=0.8em,
	column sep=1.2em
	] {
		1 & a_1 & a_2 & \dots & a_n & a_{n-1} & \dots & a_1 \\
		&     &     &       &     &         &       & \vdots \\
		&     &     &       &     &         &       & a_{n-1} \\
		&     &     &       &     &         &       & a_n \\
		&     &     &       &     &         &       & \vdots \\
		&     &     &       &     &         &       & a_2 \\
		&     &     &       &     &         &       & a_1 \\
		&     &     &       &     &         &       & 1 \\
	};
	
	\draw[decorate, decoration={brace, amplitude=5pt, raise=2pt}] 
	(m-1-2.north west) -- (m-1-5.north east);
	
	\draw[decorate, decoration={brace, amplitude=5pt, raise=2pt}] 
	(m-1-6.north west) -- (m-1-8.north east);
	
	\draw[<->, >=stealth] (m-1-1.south) -- (m-8-8.west);
	\draw[<->, >=stealth] (m-1-2.south) -- (m-7-8.west);
	\draw[<->, >=stealth] (m-1-3.south) -- (m-6-8.west);
	\draw[<->, >=stealth] (m-1-5.south) -- (m-4-8.west);
	\draw[<->, >=stealth] (m-1-6.south) -- (m-3-8.west);
	
	\node[anchor=center] at (m-6-2) {\textbf{Conjugate}};
\end{tikzpicture}.
\end{equation}

With the obvious restrictions coming from $\sigma\in Y$, this facilitates computations a lot. For example, it immediately follows that the trace-distance non-negativity of $\ket{\Y_d}$ equals to the maximum value of trace-distance measure of coherence \cite{Rana2016} for any $d$-dimensional states, for all $d\leq 4$:
\begin{equation}\label{Eq:NT.CT}
N^T_Y(\ket{\Y_d})=2\left(1-\frac{1}{d}\right)=C^T(\ket{\psi_d}),\quad d\leq 4,\quad \ket{\psi_d}:=\frac{1}{\sqrt{d}}\sum\limits_{i=1}^d \ket{i}, 
\end{equation}
and $Y=\mathcal{DDN},\mathcal{CP},\mathcal{DNN}$.
 We summarize these results in the following Table.
 
 \begin{table}[h]
 	\centering
 	\renewcommand{\arraystretch}{1.3}
 	\setlength{\tabcolsep}{12pt}
 	\begin{tabular}{c c c }
 		\hline
 		\textbf{$d$} & \textbf{$N^T_Y(\ket{\Y_d})$} & Does $N^T_Y$ differs with $Y$?\\ 
 		\hline
 		
 		$2\leq d\leq 4$ & $2\left(1-\frac{1}{d}\right)=C^T(\ket{\psi_d})$  & No \\
 		
 		$5$ & $N^T_{\text{DDN}}(\ket{\Y_d})=N^T_{\text{CP}}(\ket{\Y_d})=\frac{1}{10} \left(17-\sqrt{5}\right)$&  Yes,  $N^T_{\text{DNN}}(\ket{\Y_d})=1+\frac{1}{\sqrt{5}}$ \\ 
 		 	
 		\hline
 	\end{tabular}
 	\caption{Computation of $N^T_Y(\ket{\Y_d})$ for different dimensions $d$, and sets $Y$.}
 	\label{Tab:Y.NT}
 \end{table}

\section{Trace distance of non-negativity is not a strong monotone \label{Sec:TDNN}}
 The trace-distance of non-negativity, similar to the trace-distance of coherence \cite{Rana2016}, constitutes a valid resource measure i.e. it  vanishes on all free states and satifies both the monotonicity and convexity  conditions \cite{JOHNSTON}. However, this raises the question whether this measure satisfies the strongly monotonicity (does not increase, on average under free operation) i.e.,
\[\sum_ip_iN^T(K_i\rho K_i^\dagger/ p_i)\leq N^T(\rho)\]
where $\{K_i\}'s$ are the Kraus operators of CP-preserving quantum channel and $p_i:=\tr\left(K_i\rho K_i^\dagger\right)$. This was one of the very precise questions raised in \ccite{JOHNSTON}. We settle this question in negative by presenting an explicit counterexample.  Consider the following state 
 \[\rho= p\rho_1\oplus (1-p)\rho_2,\]
 where 
\[\rho_1=\ketbra{\Y_2}{\Y_2}=\frac{1}{2}\begin{pmatrix*}[r]
	1 &-1\\
	-1 & 1
\end{pmatrix*};\quad  \rho_2=\ketbra{\Y_3}{\Y_3}=
\frac{1}{3}\begin{pmatrix}
	1 & \omega & \omega^2\\
	\omega^2 & 1  & \omega\\
	\omega & \omega^2 &1
\end{pmatrix}.\]

Since $\idm$ commutes with every matrix, the eigenvalues of the Hermitian, trace-zero matrix $(\rho-\idm/5)$ can be calculated easily, and one verifies that 
\[\left\|\rho-\frac{\idm_5}{5} \right\|_1=\frac{1}{5} \left(| 1-5 p| +| 4-5 p| +3\right).\]
Hence, 
\[N^T(\rho)\leq \left\|\rho-\frac{\idm_5}{5} \right\|_1=\frac{1}{5} \left(|1-5 p| +| 4-5 p| +3\right). \]

Now  consider the following channel 
\[\Phi(X)=\sum_{i=1}^2K_iXK_i^\dagger,\]
where 
\[K_1=\idm_2\oplus 0 ;\quad K_2=0\oplus \idm_3.\]
Evidently, $\Phi$ is CPCP, hence CP-preserving. 
Now, \[\sum_{i=1}^2\tr\left(K_i\rho K_i^\dagger\right)N^T\left(\frac{K_i\rho K_i^\dagger}{\tr\left(K_i\rho K_i^\dagger\right)}\right)=p N^T(\rho_1)+(1-p)N^T(\rho_2)=p+(1-p)\frac{4}{3}=\frac{4-p}{3}.\]

\[\text{One verifies that }\, \frac{4-p}{3}>  \frac{1}{5}\left(| 1-5 p| +| 4-5 p| +3\right) \,\text{when} \,\frac{4}{25} < p<\frac{2}{5}.\]
 Hence
 \[\sum_{i=1}^2\tr\left(K_i\rho K_i^\dagger\right)N^T\left(\frac{K_i\rho K_i^\dagger}{\tr\left(K_i\rho K_i^\dagger\right)}\right)> N^T(\rho),\quad \text{for}\quad\frac{4}{25} < p<\frac{2}{5}.\]
 Therefore, the trace distance of non-negativity is not a strong monotone under  CPCP (and hence also for CP-preserving) free operations. While this example can trivially be extended to any higher dimension, similar examples can also be constructed for $\rho=p\ketbra{\Y_m}{\Y_m}+(1-p)\ketbra{\Y_n}{\Y_n}$, $m,n\geq 3$.
  
\printbibliography

@article{JOHNSTON,
	title = {Completely positive completely positive maps (and a resource theory for non-negativity of quantum amplitudes)},
	journal = {Linear Algebra and its Applications},
	volume = {653},
	pages = {395-429},
	year = {2022},
	issn = {0024-3795},
	doi = {10.1016/j.laa.2022.08.016},
	url = {https://www.sciencedirect.com/science/article/pii/S0024379522002932},
	author = {Nathaniel Johnston and Jamie Sikora},
}

@book{Naomi,
    author = {Shaked-Monderer, Naomi and Berman, Abraham},
	title = {Copositive and Completely Positive Matrices},
	publisher = {\href{https://doi.org/10.1142/11386 }{WORLD SCIENTIFIC}},
	year = {2021},
     address = {},
	edition   = {},
	doi = {10.1142/11386},
	eprint = {https://www.worldscientific.com/doi/pdf/10.1142/11386}
}

@book{Berman,
		author = {Berman, Abraham and Shaked-Monderer, Naomi},
		title ={Completely Positive Matrices},
		publisher = {\href{https://doi.org/10.1142/5273} {WORLD SCIENTIFIC}},
		year = {2003},
		doi = {10.1142/5273},
		address = {},
		edition   = {},
		URL = {https://doi.org/10.1142/5273},
		eprint = {https://www.worldscientific.com/doi/pdf/10.1142/5273}
}

@article{Dickinson.2014,
  title={On the computational complexity of membership problems for the completely positive cone and its dual},
  author={Dickinson, Peter JC and Gijben, Luuk},
  journal={Computational Optimization and Applications},
  volume={57},
  number={2},
  pages={403--415},
  year={2014},
  publisher={Springer},
  doi={10.1007/s10589-013-9594-z}
}

@article{Chitambar2019,
	title = {Quantum resource theories},
	author = {Chitambar, Eric and Gour, Gilad},
	journal = {Rev. Mod. Phys.},
	volume = {91},
	issue = {2},
	pages = {025001},
	numpages = {48},
	year = {2019},
	month = {Apr},
	publisher = {American Physical Society},
	doi = {10.1103/RevModPhys.91.025001},
	url = {https://link.aps.org/doi/10.1103/RevModPhys.91.025001}
}

@article{cha2026,
	title={{CPDNN} quantum channels with qubit output are {CPCP}},
	author={Cha, Hyunho},
	journal={arXiv:2603.16962},
	year={2026},
url={https://doi.org/10.48550/arXiv.2603.16962}
}

@article{Diananda.1962,
  author    = {Diananda, P. H.},
  title     = {On non-negative forms in real variables some or all of which are non-negative},
  journal   = {Mathematical Proceedings of the Cambridge Philosophical Society},
  volume    = {58},
  number    = {1},
  pages     = {17--25},
  year      = {1962},
  publisher = {Cambridge University Press},
  doi       = {10.1017/S0305004100036185}
}

@article{Johnston2019,
	title={Pairwise Completely Positive Matrices and Conjugate Local Diagonal Unitary Invariant Quantum States},
	volume={35},
	ISSN={1081-3810},
	url={http://dx.doi.org/10.13001/1081-3810.3842},
	DOI={10.13001/1081-3810.3842},
	journal={The Electronic Journal of Linear Algebra},
	publisher={University of Wyoming Libraries},
	author={Johnston, Nathaniel and MacLean, Olivia},
	year={2019},
	month=Feb, pages={156–180}
	 }

@book{2023semidefinite,
	author = {Skrzypczyk, Paul and Cavalcanti, Daniel},
	title = {Semidefinite Programming in Quantum Information Science},
	publisher ={\href{https://doi.org/10.1088/978-0-7503-3343-6} {IOP Publishing}},
	year = {2023},
	series = {2053-2563},
	isbn = {978-0-7503-3343-6},
	url = {https://doi.org/10.1088/978-0-7503-3343-6},
	doi = {10.1088/978-0-7503-3343-6}
	}

@misc{cvx,
		author       = {Michael Grant and Stephen Boyd},
		title        = {{CVX}: Matlab Software for Disciplined Convex Programming, version 2.1},
		howpublished = {\url{https://cvxr.com/cvx}},
		month        = mar,
		year         = 2014
	}

@article{Rana2016,
		title = {Trace-distance measure of coherence},
		author = {Rana, Swapan and Parashar, Preeti and Lewenstein, Maciej},
		journal = {Phys. Rev. A},
		volume = {93},
		issue = {1},
		pages = {012110},
		numpages = {7},
		year = {2016},
		month = {Jan},
		publisher = {American Physical Society},
		doi = {10.1103/PhysRevA.93.012110},
		url = {https://link.aps.org/doi/10.1103/PhysRevA.93.012110}
	}

@book{boyd2004convex,
		title={Convex optimization},
		author={Boyd, Stephen and Vandenberghe, Lieven},
		year={2004},
		publisher={\href{https://doi.org/10.1017/CBO9780511804441}{Cambridge university press}}
	}

@article{CHOI1975285,
		title = {Completely positive linear maps on complex matrices},
		journal = {Linear Algebra and its Applications},
		volume = {10},
		number = {3},
		pages = {285-290},
		year = {1975},
		issn = {0024-3795},
		doi = {https://doi.org/10.1016/0024-3795(75)90075-0},
		url = {https://www.sciencedirect.com/science/article/pii/0024379575900750},
		author = {Man-Duen Choi}
	}

@book{Watrous2018, 
		place={Cambridge},
		 title={The Theory of Quantum Information},
		  publisher={\href{https://doi.org/10.1017/9781316848142}{Cambridge University Press}}, 
		  author={Watrous, John}, year={2018}
		  }

\end{document}